\documentclass[12pt,a4paper]{article}
\usepackage{amsmath,amsfonts,amssymb,amsthm}
\usepackage{enumerate}
\usepackage{url}
\usepackage{multicol}
\usepackage{graphicx}
\usepackage{graphics}
\newtheorem{theorem}{Theorem}[section]
\newtheorem{lemma}{Lemma}[section]
\newtheorem{definition}{Definition}[section]
\title{On the structure of the set of self-similar
quadruples of point vortices in the plane
\footnote{This text is a direct translation from Polish
of a report titled
{\it O strukturze samopodobnych czworek wirow punktowych na plaszczyznie},
prepared in 2011 and published as an institutional report 
of the Faculty of Mechanical and Power Engineering,
Wroclaw University of Technology,
no 7 in the series
{\it Seria Raporty Inst. Inz. Lot. Proces. Masz. Energ. PWroc. 2011, Ser. PRE}.
}}
\author{Marek Kazimierz Lewkowicz}
\def\C{\mathbb{ C}}

\def\R{\mathbb{ R}}

\begin{document}
\maketitle
\pagestyle{headings}
\pagenumbering{arabic}
\section{Introduction and main results}

A point vortex is a point in the plane
with a non-zero real number, called its circulation, assigned to it.
This is an idealization of a vortex thread
in three-dimensional space,
which has the shape of a straight line.
Circulation defines the intensity
and orientation of rotation of the thread itself,
and consequently also of the surrounding space.
A vortex with coordinates $(x,y)\in\R^2$
(or equivalently $z=x+iy\in\C$)
and circulation $\kappa$
forces the points in the plane to move,
endowing a passive point $w\in\C$
with speed
$$i\kappa\frac{w-z}{|w-z|^2}.$$
Thus, the velocity vector at $w$
is perpendicular to the relative position vector.
Its length is inversely proportional to the distance between the points
and directly proportional to circulation,
while its orientation depends on the sign of circulation.

A finite system of point vortices in the plane
evolves in such a way
that the speed of each of the vortices
is sum of the speed caused by the other vortices,
in accordance with the above principle.

Evolution of vortex systems
has been studied for over 130 years
(see \cite{aref2010}, \cite{aref2011}, \cite{kudela}).
As early as 1883 Kirchhoff showed
that this movement can be described
using the formalism of Hamilton.
Quite surprisingly,
there are systems of three vortices
which during evolution
collapse to a point.
This result was already included
in the Groebli's work \cite{grobli} from 1877.
In 1979 Novikov \cite{novikov}
gave examples of systems of $n$ vortices 
collapsing to a point for $n=4$ and $5$.
His example (for $n=4$) concerns in fact
only one circulation
(see Section \ref{section_novikov} below)
and the set of collapsing systems given by him 
(after a natural completion)
can be identified with a circle in the plane.
Therefore it constitutes a single smooth simple closed curve.
In 1987, O'Neil \cite{oneil} proved
that collapsing systems exist for any $n$.
His proof works for
a variety of circulations.
O'Neil was able to prove
that at least some connected components of the set of self-similar systems 
are curves, although one cannot rule out
isolated singularities or self-intersections.
This is due to the fact that the methods used by him
belong to algebraic geometry.
They imply that the set is a one-dimensional
algebraic variety.

The aim of this study is to take a closer look at
the collection of self-similar systems
for the circulation used by Novikov.
It should be noted that O'Neil's results
do not apply directly to this particular circulation.
We introduce a numerical procedure searching for self-similar systems
and we shall look at the set it finds.
The resulting image suggests
that for this circulation
the set has five algebraic components,
each of which is a smooth
simple closed curve.
The curves seem to have
neither singularities nor self-intersections,
although one of them seems to interesect two others
in single points, so that the number of connected
components could be three.

The set found by Novikov is one of these components.
In this work we numerically find four other components.
They differ from the Novikov component 
in that they are not flat circles
in any two-dimensional affine space.
In the final part of the work 
we refer to the results of O'Neil
and point out that the existence of the components found here
is not predicted by these results.

The results below are mainly numerical.
The algebraic equation proposed in the work,
describing the set of self-similar systems,
is derived rigorously.
However, any application of the equation
to a given vortex problem
gives a numerical solution
subject to numerical error.
Therefore the results of the paper
suggesting that the set of solutions 
is a union of five smooth curves
requires more rigorous
mathematical reasoning.
We put it down to another study.

\section{Layout of the work}
The next section contains basic definitions
and facts, taken from the work of O'Neil.
I give the definition of the dynamical system
which is the subject of our interest.
I recall the concept of self-similar systems of vortices
and distinguish four specific situations,
in which a system may be invariant, translational, rotational, 
or collapsing.
I present some basic invariance properties of the problem,
related to the action of the group of affine transformations in the plane.
In particular this leads to the concept of configuration.
Configuration space has dimension
lower by four than that of the phase space.

In the following part of the work
I propose to study the set of the self-similar configurations
by representing it as a set of zeros
of some vector field $U$
(see definition \ref{pole_u}).
I prove that in fact the zeros of $U$
coincide with the self-similar configurations.
Next I propose a numerical procedure
searching for self-similar configurations
by seeking zeros of the field $U$ 
with the gradient method applied to the function $|U|^2$.
We want to apply this method to systems of four vortices.
In order to test it, we formulate it directly for three vortices.
We present the image of the set 
of self-similar points obtained with this method
for three vortices
and compare it with the well-known solution, 
which is a circle.

The next section refers to the Novikov examples 
of collapsing systems for $n=4$.
We check that the configurations described by him
form a flat circle.
Next I give the results of my numerical calculations.
I present graphic images of plane curves
obtained as projections to the coordinate system planes
of the set of self-similar configurations
for the circulation used by Novikov.
The set of configurations was obtained numerically
as the set of zeros of the vector field $U$.
In these images one can easily see five connected components.
One of them is isometric to a circle - that's the Novikov component.
The other four are the main goal of this paper.

\section{Basic definitions and facts}
For a natural number $n$,
a circulation is a sequence of non-zero real numbers
$\kappa=(\kappa_1..\kappa_n)\in\R_*^n$,
and a discrete system of $n$ point vortices
(briefly a {\it system}) is a point of the phase space
$$DV^{[n]}=\{z=(z_1,\dots,z_n)\in\C^n:\
\forall_{j\neq k}\ 
z_j\neq z_k\}.$$
According to the commonly accepted definition
(see e.g. O'Neil \cite[p. 384, def. 0.1]{oneil}, Aref \cite{aref2010}),
we consider a dynamical system
in the phase space $DV=DV^{[n]}$
defined by the vector field
$$V=V^{[n]}=(V_1,\dots,V_n),\quad
V_j=\left(\sqrt{-1}\right)
\sum_{k=1,k\neq j}^n
\kappa_k\frac{z_j-z_k}{|z_j-z_k|^2}.$$
Therefore we are interested in the trajectories of the field, 
i.e., the curves 
$t\to z(t)\in DV^{[n]}$ 
satisfying
$$z'(t)=V^{[n]}(z(t)).$$
As we know from the elementary theory of dynamical systems,
for any initial conditions
$z(0)\in DV$ 
there exists exactly one such curve,
defined in maximal domain
being an interval, bounded or not.

We distinguish the following four types of vortex systems $Z\in DV$
depending on the value of $V=V^{[n]}(Z)$.
(cf. O'Neil \cite[def. 1.1.3, p. 387]{oneil})
\begin{itemize}
	\item $Z$ is a {\it fixed point} (equilibrium) if $V_1=\dots V_n=0$.
	\item $Z$ is a {\it translational} system
if $V_1=\dots V_n=v\neq0$ for some $v\in C$. 
	\item $Z$ is {\it rotational} (relative equilibrium)
if some constants $0\neq\lambda\in R$, 
$z_0\in C$, satisfy $\forall_l\ V_l=i\lambda(Z_l-z_0)$.
	\item $Z$ is a {\it collapsing} system
if for some constants $0\neq\omega, z_0\in C$ 
we have $\forall_l\ V_l=\omega(Z_l-z_0)$
and $Re(\omega)\neq0$.
\end{itemize}

A system $Z$ is called {\it self-similar}
if, during evolution, it remains similar to the initial state.
In other words, if $z(0)=Z, z'(t)=V(z(t))$,
then for any $t$
the system $z(t)$ can be obtained from $Z$
by dilation composed with isometry.
It is easy to prove (see e.g. O'Neil \cite{oneil}), 
that $Z$ is self-similar
if and only if there is a constant $\omega\in C$,
for which
$$\forall_{j,k}\ 
(V_k-V_j)=\omega(z_k-z_j),$$
and that in turn holds if and only if
the system $Z$ belongs to one of the above listed,
mutually exclusive types.

\begin{lemma}
If a system $Z$ is collapsing
in the sense of the above definition
(i.e., $V_l(Z)=\omega(Z_l-z_0)$,
$Re(\omega)\neq0$),
then it is self-similar
and its temporal evolution is described by the formula
$$z(t)=z_0+\sqrt{2 Re(\omega)t+1}\
e^{i\frac{Im(\omega)}{2 Re(\omega)}\ln(2 Re(\omega)t+1)}(Z-z_0).$$
For a rotational system, with $V_l=i\lambda(Z_l-z_0)$, we have
$$z(t)=z_0+e^{i\lambda t}(Z-z_0).$$
Thus a rotational system rotates isometrically
around $z_0$ at a constant speed,
while for a collapsing system
all distances between the points
$z_k(t)$ tend to zero
when $t\to t_0=-\frac{1}{2 Re(\omega)t}$.
\end{lemma}

Notice that $t_0>0$ if $Re(\omega)<0$
and the system collapses to a point in finite time.
If $Re(\omega)>0$ then $t_0<0$.
The system collapses to a point during negative
(reversed) flow of time,
while for the positive flow of time 
the system expands to infinity.

\begin{proof}
Due to the fact that the field $V$ is translation invariant,
we can address only $Z$ such that
$$V(Z)=\omega Z.$$
I will show that there are real functions $r(t),\phi(t)$
satisfying $r(0)=1,\phi(0)=0$,
for which the curve $z(t)=r(t)e^{i\phi(t)}Z$
is an integral curve.
Due to the form of the field $V$ we have
$$V(z(t))
=\frac{e^{i\phi(t)}}{r(t)}V(Z)
=\omega\frac{e^{i\phi(t)}}{r(t)}Z.$$
We substitute this expression to the equations of motion $z'(t)=V(z(t))$
and we get
$$r'(t)+r(t)i\phi'(t)=\frac{\omega}{r(t)}.$$
It follows that
$$r'(t)=\frac{Re(\omega)}{r(t)},
\quad r^2(t)\phi'(t)=Im(\omega).$$
Of course, if $Re(\omega)=0$,
then $r(t)$ is constantly equal 1, and $\phi(t)=Im(\omega)t$.
If $Re(\omega)\neq0$, then the first equation gives
$$r^2(t)=2 Re(\omega)t+1,$$
and from the second one we have
$\phi'(t)=\frac{Im(\omega)}{2 Re(\omega)t+1}$,
so that
$$\phi(t)=\frac{Im(\omega)}{2 Re(\omega)}\ln(2 Re(\omega)t+1).$$
\end{proof}
There is a long tradition to consider self-similar systems.
One of the arguments in favor of this
is that for such systems
we can explicitly solve the equations of motion.
Another argument is as follows.
One can consider systems which collapse 
in a more general sense,
namely systems for which
the distance of two or more vortices tends to zero
in finite time during the evolution.
In the paper published in 2007 (preprint 2004)
Garduno and Lacombe \cite{garduno} showed for $n=3$
that every system collapsing in the broader sense
is self-similar,
therefore is collapsing in the sense adopted here.
The question is still open for systems of four or more vortices.

I will now discuss a few properties of self-similar systems.
This will facilitate presentation of results
and make it more transparent.
Let me start with the invariance with respect to similarities.
The group $H$ of the orientation preserving similarities of the plane
(that is, linear transformations conserving angles,
or, in other words, compositions of translations, rotations and dilations)
is a four-dimensional connected, non-compact Lie group.
Transformations belonging to $H$ 
are invertible $\C$-linear maps 
$$C\ni z\to az+b\in C,
\quad a,b\in C, a\neq0.$$
Note that this group acts freely and 2-transitively on $\C$.
In other words, for any numbers $z_1,z_2,w_1,w_2\in\C$
satisfying $z_1\neq z_2$ and $w_1\neq w_2$
there exists precisely one $\phi\in H$ such that
$\phi(z_k)=w_k,k=1,2$.
Thus, the natural action of $H$ on $DV^{[n]}$ for $n\ge2$
given by $\phi(z)_k=\phi(z_k)$
is free: if $\phi(z)=z$ for some $z$,
then $\phi=I$.
It follows that the orbits of the group $H$ 
in the phase space are four dimensional submanifolds,
form a foliation, and the quotient space $C^{[n]}=DV^{[n]}/H$,
called the configuration space,
is a manifold of dimension $2n-4$.
Next, the operation preserves the type of points
in the sense of the above definition.
For example, if $Z$ is a collapse,
i.e., it satisfies $V_k(Z)=\omega(Z_k-z_0)$,
then $W=\phi(Z)=aZ+b$ is also a collapse, namely
$$V_k(W)=V_k(aZ+b)
=\frac{a}{|a|^2}V_k(Z)
=\frac{a}{|a|^2}\omega(Z_k-z_0)=$$
$$=\frac{\omega}{|a|^2}(aZ_k-az_0)
=\frac{\omega}{|a|^2}(W_k-az_0-b).$$

Note that the constant $\omega$ 
has been rescaled by a
positive coefficient $\frac{1}{|a|^2}$.
Thus the phase is conserved.
More generally, one can show,
that if $W=\phi(Z)=aZ+b$
then the trajectories
emanating from the points $Z$ and $W$, respectively,
are coupled by $\phi$:
$$W(t)=\phi\left(Z\left(\frac{1}{|a|^2}t\right)\right).$$
It follows that it makes sense to speak
of various types of configurations,
and in particular of rotating
or collapsing configurations.

A related kind of invariance is contained in the following:
if $t\to z(t)$ 
is an integral curve for the circulation $\kappa$
and $c\in\R_*$ then $t\to w(t)=z(ct)$
is an integral curve for the circulation $c\kappa$.

In the sequel we shall be interested
(for a given circulation $\kappa$)
in the set of self-similar collapsing and rotating configurations
$$S_\kappa=\{[Z]:
\exists_{0\neq\omega,z_0}\
V(Z)=\omega(Z-z_0)\}.$$
Here, $[Z]\in\C^{[n]}$ denotes the configuration
corresponding to the system $Z$,
i.e., the set of all systems similar to $Z$.

For natural $n$, let $DU^{[n]}$
denote the following open subset of $\C^n$
$$DU^{[n]}=\{w\in\C^{n}:
\forall_k(w_k\neq 0,1),\forall_{j\neq k}\ (w_k\neq w_j)\}.$$
Notice that each orbit of the group $H$ 
acting on $DV^{[n]}$
contains precisely one point
of the form $(w,1,0)$ and $w\in DU^{[n-2]}$.
It follows that the configuration space for
$n$-vortices can be naturally identified
with the $(2n-4)$-dimensional space $DU^{[n-2]}$.

For future use, let us introduce the following vector field
\begin{definition}\label{pole_u}
For natural $n$ and for $\kappa\in\R_*^{n+2}$,
we define on the set $DU^{[n]}$
the vector field $U=U^{[n]}$ by the formula
$$U_k(w)=(\kappa_k+\kappa_{n+2})\frac{w_k}{|w_k|^2}
+\kappa_{n+1}\left(1+\frac{w_k-1}{|w_k-1|^2}\right)
-w_k\kappa_k\left(\frac{1-w_k}{|1-w_k|^2}+\frac{w_k}{|w_k|^2}\right)
-$$
$$-w_k(\kappa_{n+1}+\kappa_{n+2})
+\sum_{j\le n,j\neq k}\kappa_j\left(
\frac{w_k-w_j}{|w_k-w_j|^2}+\frac{w_j}{|w_j|^2}
-w_k\left(\frac{1-w_j}{|1-w_j|^2}+\frac{w_j}{|w_j|^2}\right)\right)
.$$
\end{definition}

\section{Self-similar systems as zeros of a vector field}
In this section I will prove
that the task of determining all self-similar systems
is equivalent to determination of all zeros
of the vector field $U$ introduced above.

For the curve $z(t)$ in the space $DV^{[n]}$
we consider the family of affine transformations
$$\C\ni z\to\phi_t(z)
=\frac{z-z_n(t)}{z_{n-1}(t)-z_n(t)}\in\C.$$
If each point $z(t)$
is transformed by the similarity $\phi_t$ (depending on $t$),
we get a curve, whose last two coordinates
are always equal to 1 and 0. 
It is therefore a curve of the form $(w(t),1,0)$,
where
$$w_k(t)=\phi_t(z_k(t))=\frac{z_k(t)-z_n(t)}{z_{n-1}(t)-z_n(t)}$$
for $k\le n-2$. The curve $w(t)$ has values in $DU^{[n-2]}$.

Let $\kappa=(\kappa_j)_{j=1..n}$ be a fixed circulation
and $U=U^{[n-2]}$ be the vector field on $DU^{[n-2]}$ defined above.
\begin{theorem}
If $z(t)$ is an integral curve of the field $V^{[n]}$
then the curve $t\to w(t)\in DU^{[n-2]}$ satisfies the equation
$$w_k'(t)
=\frac{iU(w(t))}{|z_{n-1}(t)-z_n(t)|^2}.$$
\end{theorem}
\begin{proof}
Denote $z_{jk}=z_j-z_k$ and $\zeta=z_{n-1}-z_n$.
Since $z(t)$ is an integral curve of $V$,
$$z_k'=i\sum_{j\neq k}\kappa_j\frac{z_{kj}}{|z_{kj}|^2}.$$
Hence
$$-iz_{kp}'
=(\kappa_k+\kappa_p)\frac{z_{kp}}{|z_{kp}|^2}
+\sum_{j\neq k,p}\kappa_j
\left(\frac{z_{kj}}{|z_{kj}|^2}
+\frac{z_{jp}}{|z_{jp}|^2}\right).$$
Furthermore $w_k=z_{kn}/\zeta$, so that
$$-i\zeta^2w_k'
=-i\zeta^2\frac{z_{kn}'\zeta-z_{kn}\zeta'}{\zeta^2}
=-iz_{kn}'\zeta-z_{kn}(-i\zeta')=$$
$$=\left(
(\kappa_k+\kappa_n)\frac{z_{kn}}{|z_{kn}|^2}
+\sum_{j\neq k,n}\kappa_j
\left(\frac{z_{kj}}{|z_{kj}|^2}
+\frac{z_{jn}}{|z_{jn}|^2}\right)\right)\zeta+$$
$$-z_{kn}\left((\kappa_{n-1}+\kappa_n)\frac{\zeta}{|\zeta|^2}
+\sum_{j\le n-2}\kappa_j\left(\frac{z_{n-1,j}}{|z_{n-1,j}|^2}
+\frac{z_{jn}}{|z_{jn}|^2}\right)\right).$$
Npw we use
$$z_{kn}=w_k\zeta,
\quad\quad z_{kj}=z_{kn}-z_{jn}=(w_k-w_j)\zeta.$$
Therefore
$$-i|\zeta|^2w_k'=(\kappa_k+\kappa_n)\frac{w_k-w_n}{|w_k-w_n|^2}
+\sum_{j\neq k,n}\kappa_j\left(\frac{w_k-w_j}{|w_k-w_j|^2}
+\frac{w_j-w_n}{|w_j-w_n|^2}\right)+$$
$$-(w_k-w_n)\left(\kappa_{n-1}+\kappa_n
+\sum_{j\le n-2}\kappa_j\left(\frac{w_{n-1}-w_j}{|w_{n-1}-w_j|^2}
+\frac{w_j-w_n}{|w_j-w_n|^2}\right)\right)=$$
$$=(\kappa_k+\kappa_n)\frac{w_k}{|w_k|^2}
+\kappa_{n-1}\left(1+\frac{w_k-1}{|w_k-1|^2}\right)
-w_k\kappa_k\left(\frac{1-w_k}{|1-w_k|^2}+\frac{w_k}{|w_k|^2}\right)-
$$
$$-w_k(\kappa_{n-1}+\kappa_n)+\sum_{j\le n-2,j\neq k}\kappa_j\left(
\frac{w_k-w_j}{|w_k-w_j|^2}+\frac{w_j}{|w_j|^2}
-w_k\left(\frac{1-w_j}{|1-w_j|^2}+\frac{w_j}{|w_j|^2}\right)\right)
.$$
\end{proof}

As I mentioned previously,
if a self-similar system is transformed by a similarity, 
we get a self-similar system again.
In other words, 
self-similarity is a characteristic property
and can be tested in the configuration space.
It follows that in order to determine
all self-similar systems,
it is enough to determine those self-similar systems $z$,
for which $z_{n-1}=1$, $z_n=0$.
Each of the self-similar systems found
will give a four space of self-similar systems,
namely the orbit under the action of the group $H$.

In the remaining part of this section
I shall therefore deal only with self-similar systems
$z\in DV^{[n]}$ of the form $z=(w,1,0)$,
$w\in DU^{[n-2]}$.

\begin{theorem}\label{tw_o_zerach}
For $w=(w_1,\dots w_{n-2})\in DU^{[n-2]}$,
a system of $n$ vortices $z=(w,1,0)\in DV^{[n]}$
is self-similar if and only if 
$w$ is a zero of the vector field $U$ 
defined in the previous section.
\end{theorem}
\begin{proof}
The trajectory $z(t)$ of the field $V$
emerging from $z=(w,1,0)$ is self-similar
if and only if the curve $(w(t),1,0)\in C^n$,
defined in the proof above, is self-similar.
However, any two systems of the form $(w(t),1,0)$
have two coordinates equal,
so they are similar only if the trajectory 
$w(t)\in C^{n-2}$
is constant.
According to the formula form the Theorem,
this occurs when $w$ is a zero of the field $U$,
because $z_n(t)-z_{n-1}(t)\neq0$.
\end{proof}

\section{Systems of three vortices}\label{trzywiry}
The above theorem will be applied for $n=4$,
but for control purposes we apply it to $n=3$,
where the result is well known.

As we know, self-similar collapsing systems
may exist only for circulations $\kappa$, 
for which $L=\sum_{l<j}\kappa_l\kappa_j=0$.
Taking homogeneity into account, we can consider only circulations
with $\kappa_1=-1$.
For the circulation $\kappa=(-1,\kappa_2,\kappa_3)$
the condition $L=0$ gives
$$\kappa_2\kappa_3=\kappa_2+\kappa_3.$$
Therefore, circulations under consideration
depend only on one parameter $\lambda$
and one can take
$$\kappa_1=-1,\quad \kappa_2=\lambda,
\quad \kappa_3=\frac{\lambda}{\lambda-1}.$$
Our claim is that the point $(w,1,0)$
is self-similar if and only if it nullifies the expression
$$U(w)=(\kappa_1+\kappa_3)\frac{w}{|w|^2}
+\kappa_2\left(1+\frac{w-1}{|w-1|^2}\right)
-w\kappa_1\left(\frac{1-w}{|1-w|^2}+\frac{w}{|w|^2}\right)
-w(\kappa_2+\kappa_3)=$$
$$=\frac{1}{\lambda-1}\frac{w}{|w|^2}
+\lambda\left(1+\frac{w-1}{|w-1|^2}\right)
+w\left(\frac{1-w}{|1-w|^2}+\frac{w}{|w|^2}\right)
-w\frac{\lambda^2}{\lambda-1}.$$
Myltiplying both sides by
$(\lambda-1)\bar{w}(\bar{w}-1)$
we get
$$\bar{w}-1
+\lambda(\lambda-1)\bar{w}^2
-(\lambda-1)w
-w\bar{w}(\bar{w}-1)\lambda^2=0.$$
Now we substitute $w=x+yi$:
$$x-yi-1+\lambda(\lambda-1)(x^2-y^2-2xyi)
-(\lambda-1)(x+yi)-(x^2+y^2)(x-yi-1)\lambda^2=0$$
and take the imaginary part:
$$-y-2\lambda(\lambda-1)xy
-(\lambda-1)y+(x^2+y^2)y\lambda^2=0.$$
After dividing by $y\lambda^2$, this gives
$$(x^2+y^2)-2x\frac{\lambda-1}{\lambda}-\frac{1}{\lambda}
=\left(x-\frac{\lambda-1}{\lambda}\right)^2+y^2
-\left(\frac{\lambda-1}{\lambda}\right)^2-\frac{1}{\lambda}
=0.$$
Therefore
$$\left(x-\frac{\lambda-1}{\lambda}\right)^2+y^2
=\left(\frac{\lambda-1}{\lambda}\right)^2+\frac{1}{\lambda}
=\frac{\lambda^2-\lambda+1}{\lambda^2}.$$
This is an equation of a circle, 
since the determinant of the quadratic function
$\lambda^2-\lambda+1$ is negative.

Similarly, the real part
$$x-1+\lambda(\lambda-1)(x^2-y^2)
-(\lambda-1)x-(x^2+y^2)(x-1)\lambda^2$$
can be decomposed 
$$(\lambda x^2-2(\lambda-1)x+\lambda y^2-1)(1-\lambda x).$$
Finally, the solution set consists
of the point $(\frac{1}{\lambda},0)$
and a circle centered at $(\frac{\lambda-1}{\lambda},0)$
with radius
$\sqrt{\frac{\lambda^2-\lambda+1}{\lambda^2}}$.
This is consistent with the formulas in Aref's work \cite{aref2010}.
He says that if two vortices a placed at the points
$$\left(\frac{\kappa_3}{\kappa_3+\kappa_2},0\right),
\quad
\left(-\frac{\kappa_2}{\kappa_3+\kappa_2},0\right),$$
then the third one is in the circle
with center $(0,0)$ and radius
$\sqrt{\frac{\kappa_1+\kappa_2+\kappa_3}{\kappa_2+\kappa_3}}$.
Substitution of our circulations gives
the vortices positions
$$\left(\frac{1}{\lambda},0\right),
\quad\left(-\frac{\lambda-1}{\lambda},0\right),$$
and the radius
$$\sqrt{\frac{\lambda^2-\lambda+1}{\lambda^2}}.$$
It turns out that the radii are the same,
but the two vortex positions obtained by Aref
as well as the center of the circle
are translated with respect to ours
by a vector $\left[1-\frac{1}{\lambda},0\right]$.

It is well known
that the circle contains both collapsing and rotating points,
while the isolated point is a fixed point.

\section{A numerical procedure}\label{procedura}

It is clear from Theorem \ref{tw_o_zerach}
that in order to find self-similar rotating and collapsing
systems of $n$ vortices,
we have to find 
zeros of the vector field $U$ defined
in the space
$DU^{[n-2]}$ 
of dimension $2n-4$
by Definition \ref{pole_u}.
We shall use a simple numerical approach.

Zeros of the vector field $U$
are zeros (and hence the absolute minima)
of the non-negative smooth function
$$DU\ni w\to f(w)=|U(w)|^2\in\R^+.$$
We search the minima with a simple gradient method.
The parameters of the procedure are
the starting point $P$, the initial time step $dt$,
the desired accuracy $\epsilon$, and the maximum number of steps $N$.
The procedure calculates repeatedly (in a loop)
the gradient at the current point $p$,
the shifted point $q=p-\nabla fdt$,
and the value of the function at $q$.
This value is compared with the value at the point $p$.
If the new value is smaller,
then $q$ becomes the current point.
If not, then we halve the step.
We stop when the value of the function
falls below a predetermined error $\epsilon$
or after the maximum number of steps $N$ has been reached.
If, after the shutdown of the procedure,
the function value is not less than $\epsilon$,
we believe that the point we found
is a minimum which is not an absolute minimum 
(i.e., not a zero of the vector field $U$)
and therefore the point is discarded.

The gradient can be calculated 
by approximating the derivative 
by the difference ratio
$$\frac{df}{dx}
\approx\frac{f(x+dx)-f(x)}{dx}$$
for a small increment $dx$,
but if the function is a polynomial
(or, more generally, is analytic)
it is better to use the derivative
calculated analytically.
This is particularly important
when we are near the minimum point.

The procedure is applied as follows.
In some preferred region of $DU$
(e.g., in a box), we choose 
a large number of random starting points.
For each of them the gradient procedure is executed.
The absolute minima found form a discrete collection, 
whose graphic image is supposed
to provide information on the properties
of the set of self-similar rotating and collapsing systems.

\section{Graphic images of systems for n=3}
For $n=3$ the set of self-similar systems is a well known circle,
as discussed above in Example \ref{trzywiry}.
The set can be obtained numerically
by searching zeros of the vector field indicated above.

\begin{figure}[ht]
\caption{Gradient lines of $|U(w)|^2$ 
convergent to the self-similar collapsing configurations of three vortices.
The gradient is calculated numerically.}
\includegraphics[width=1.00\textwidth]{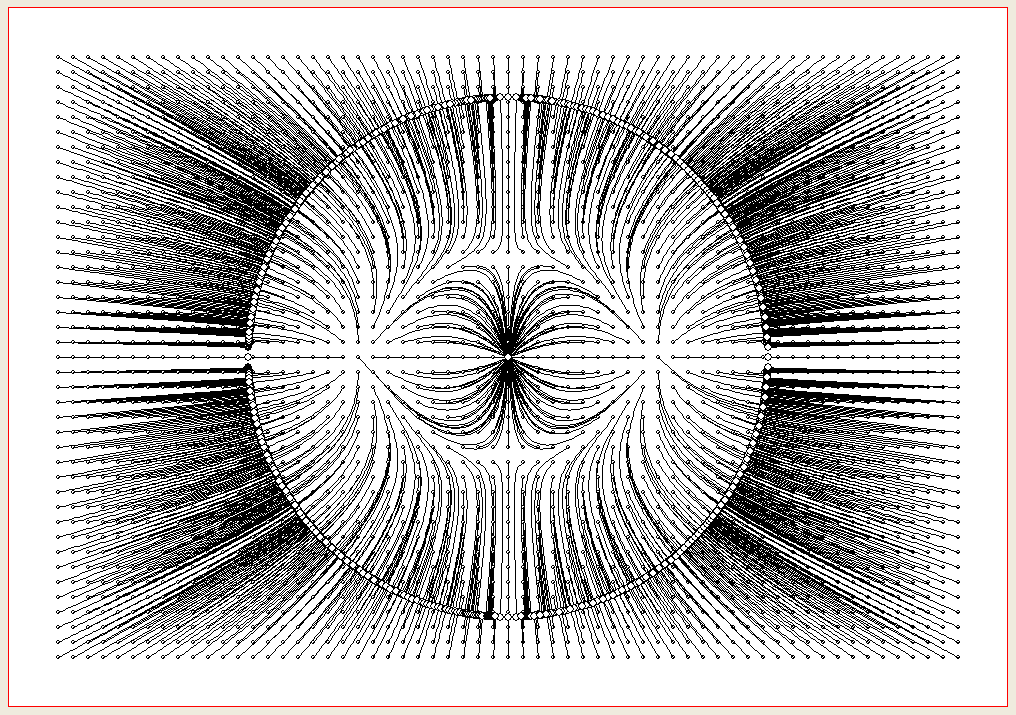}
\end{figure}

The calculations presented graphically in Fig. 1
are performed for the circulation $\kappa=(-1,2,2)$.
According to \cite{aref2010} (as well as to the calculations
from Section \ref{trzywiry}),
if two vortices are placed at the points $(0,0)$ and $(1,0)$,
then the third one point in the circle
with center $(1/2,0)$ and radius $\sqrt{3}/2\approx0.866$.
Thus, it intersects the $Ox$ axis at points $(-0.366)$ i $(1.386,0)$.

Fig. 1 contains the image of the gradient lines
of the function $|U(w)|^2$ in this case.
The gradient was calculated as
the difference ratio with $dx=10^{-6}$.
Starting points for the gradient procedure
(indicated in the figure with small circles)
are selected from the rectangle $[-1,2]\times[-1,1]$
in a regular grid of rectangular step 0.05.
Gradient lines originating at these points
converge to the minima
(slightly larger circles).
Except for the fixed point $(0.5,0)$,
they are rotating and collapsing systems
lying on the said circle.

\begin{figure}[ht]
\caption{Gradient lines of $|U(w)|^2|w|^2|w-1|^2$
convergent to the self-similar collapsing configurations of three vortices.
The gradient is calculated analytically.}
\includegraphics[width=1.00\textwidth]{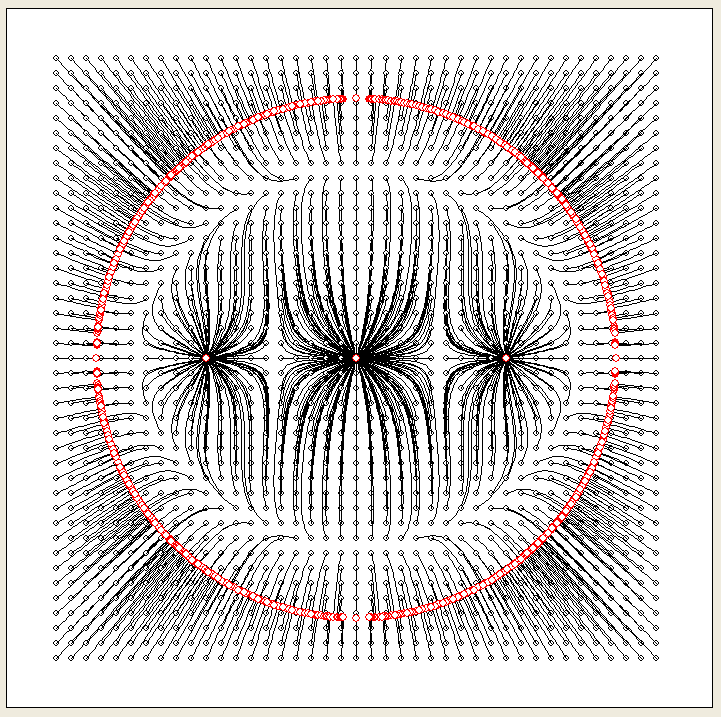}
\end{figure}

Fig. 2 shows a similar picture 
for the same circulation, step, and data grid.
The main difference lies in the fact that the minimum was counted
for the function $|U|^2$ multiplied by $|w|^2|w-1|^2$.
Therefore, the graph contains two false minima
at $(0,0)$ and $(1,0)$.
Other minima (absolute) are unchanged
because multiplying the studied function by a positive function
does not change its zeros.
In this representation, the minimized function is a polynomial
and the gradient was calculated from an analytical formula.

\section{Novikov examples for n=4}\label{section_novikov}
In 1979, Novikov gave in \cite{novikov} examples 
of collapsing systems for $n=4$.
In these examples, four points must be
the vertices of a parallelogram.
Vorticity of opposite corners must be equal.

\begin{figure}[ht]
\caption{Projection of the set of self-similar 
configurations of four vortices to the 1-2 plane.}
\includegraphics[width=0.80\textwidth]{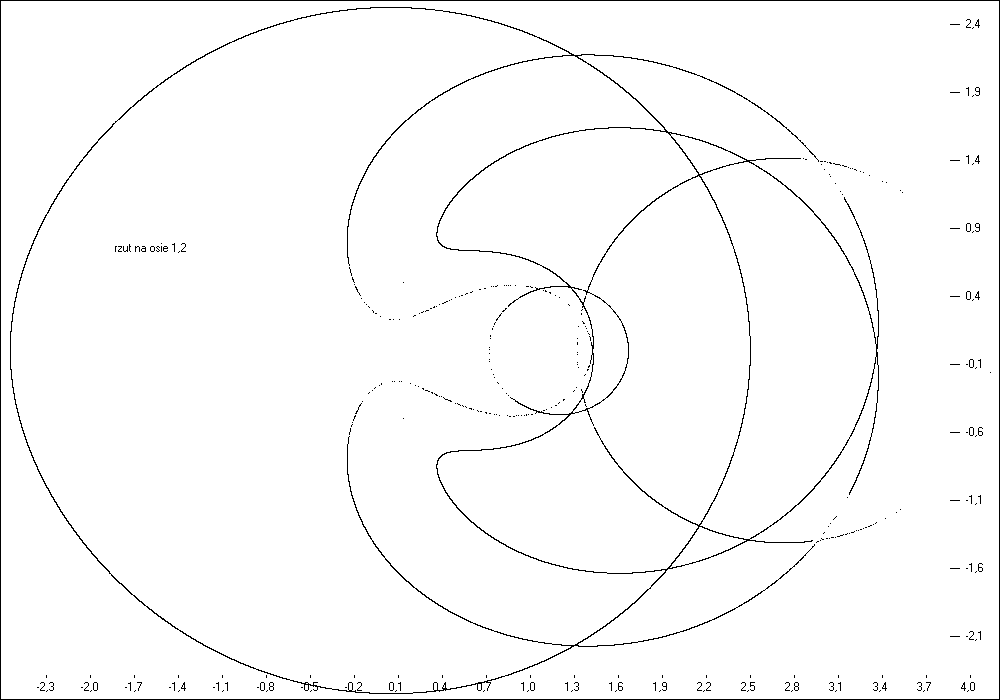}
\end{figure}

\begin{figure}[ht]
\caption{Projection of the set of self-similar 
configurations of four vortices to the 1-3 plane.}
\includegraphics[width=0.80\textwidth]{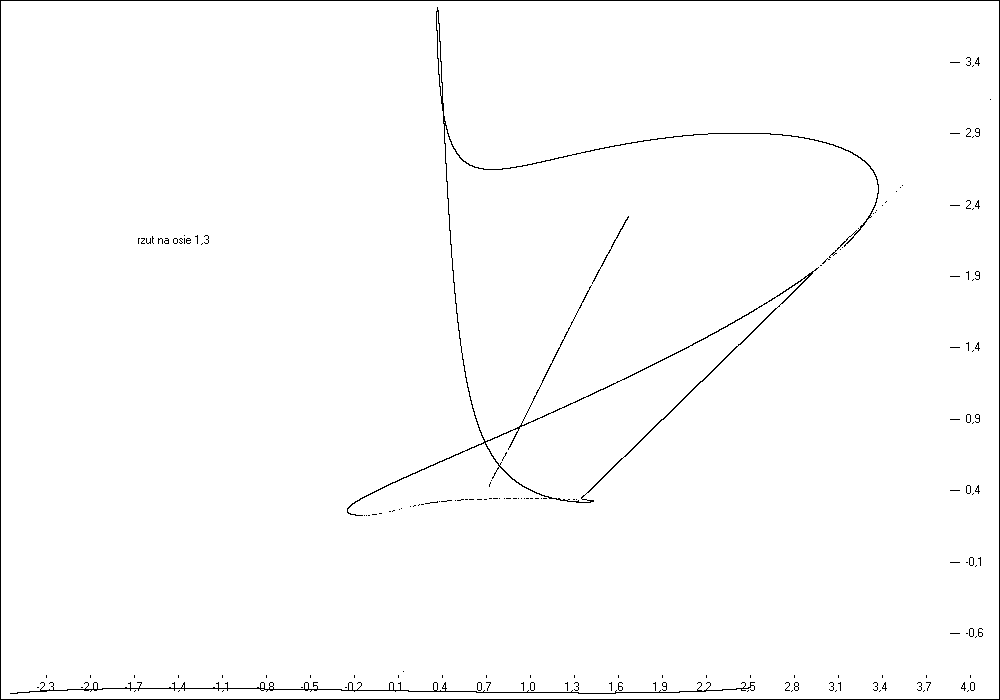}
\end{figure}

\begin{figure}[ht]
\caption{Projection of the set of self-similar 
configurations of four vortices to the 1-4 plane.}
\includegraphics[width=0.80\textwidth]{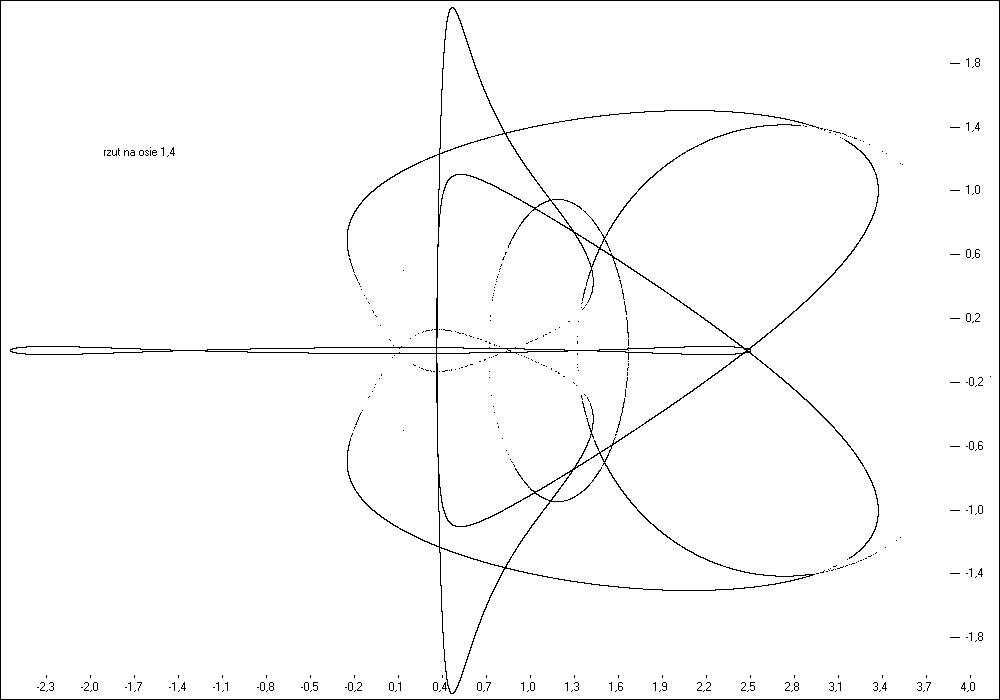}
\end{figure}

Thanks to the invariance under the action of similarities
and the homogeneity of vorticity, we can assume
that the points and their vorticity are as follows
$$(w+1,w,1,0)\in\C^4,
\quad (1,\lambda,\lambda,1)\in\R_*^4.$$
The necessary condition for
the existence of collapsing systems
$$L=\sum\kappa_i\kappa_j=\lambda^2+4\lambda+1=0$$
has to be satisfied. Therefore $\lambda$
can take on one of two values $\lambda_\pm=-(2\pm\sqrt3)$.

According to Novikov, the sum of products 
of squared lengths of the diagonals
multiplied by the intensity of any end-point of the diagonal must be zero.
Here, this means that
$$1\cdot|(w+1)-0|^2+\lambda\cdot|w-1|^2=0.$$
Thus the ratio of the length of the diagonal squares
should be $-\lambda$:
$$\frac{|w+1|^2}{|w-1|^2}=-\lambda.$$

For $\lambda=\lambda_+=-(2+\sqrt3)$ this means that
$$\frac{|w+1|^2}{|w-1|^2}=2+\sqrt3.$$
Denoting $w=x+iy$ we have
$$(2+\sqrt3)((x-1)^2+y^2)=(x+1)^2+y^2,$$
$$(1+\sqrt3)(x^2+y^2+1)=2(3+\sqrt3)x,$$
$$x^2+y^2+1
=2\frac{3+\sqrt3}{1+\sqrt3}x
=2\frac{(3+\sqrt3)(\sqrt3-1)}{2}x
=2\sqrt3 x,$$
$$(x-\sqrt3)^2+y^2=2.$$

The set of solutions
in the four-dimensional configuration space
is isometrically an ellipse 
lying in the plane 
$w_1=w_2+1$.
Its projection onto each of the two complex axes 
$Ow_1$, $Ow_2$,
is a circle.
This ellipse can be described parametrically by
$$\theta\to(1+\sqrt3+\sqrt2e^{i\theta},\sqrt3+\sqrt2e^{i\theta}).$$

Let us look at the set of self-similar configurations 
for the second value, i.e., for
$\lambda=\lambda_-=-(2-\sqrt3)$.
The condition
$$\frac{|w+1|^2}{|w-1|^2}=2-\sqrt3$$
is equivalent to 
$$2+\sqrt3=\frac{|w-1|^2}{|w+1|^2}
=\frac{|(-w)+1|^2}{|(-w)-1|^2}.$$
Thus, the set of solutions is a circle
homothetic to the previous one.

One can look at this relationship as follows.
Multiplying a circulation by a constant
does not change the set of self-similar configurations,
and moreover $\lambda_+\lambda_-=1$,
therefore the following circulations are equivalent
(they have the same set of self-similar configurations):
$$(1,\lambda_-,\lambda_-,1)
=(1,\lambda_+^{-1},\lambda_+^{-1},1)
\approx(\lambda_+,1,1,\lambda_+).$$

Permutation group of four elements
acts on vortex systems (positions and circulations)
and conserves their dynamics.
In order to obtain self-similar configurations
for the circulation $(1,\lambda_+,\lambda_+,1)$
one can treat the system with a permutation
replacing the first two coordinates
and at the same time the last two coordinates.
Then the system $(w+1,w,1,0)$
will transform into $(w,w+1,0,1)$.
It can be represented by a system
having 1 and 0 on the last two coordinates
if we act with a similarity $z\to 1-z$.
We receive $(-w+1,-w,1,0)$,
which shows that indeed the circle 
transforms into a homothetic circle.

The above remarks about symmetry occurring here
can be further developed.
Circulations of the form $(1,\lambda,\lambda,1)$
admit a four-element invariance group $G=Z_2\times Z_2$.
It acts on self-similar configurations
in such a way that one generator, say $\phi\in G$,
swaps the first and last coordinate of the system,
while the second generator --- $\psi\in G$ --- swaps the other two coordinates.
In order to identify this action on the Novikov circle
we take $u=(w+1,w,1,0)$
and replace the images
$\phi(u)=(0,w,1,w+1)$ 
and $\psi(u)=(w+1,1,w,0)$
with similar systems,
having 1 and 0 on the last coordinates.
The needed similarities are
$z\to\alpha(z)=1-\frac{z-1}{w}$ and
$z\to\beta(z)=\frac{z}{w}$.
Thus
$$[\phi(u)]=[\alpha(\phi(u))]
=\left[\left(\frac{1}{w}+1,\frac{1}{w},1,0\right)\right],$$
and 
$$[\psi(u)]=[\beta(\psi(u))]
=\left[\left(\frac{1}{w}+1,\frac{1}{w},1,0\right)\right].$$
It turns out that in terms of $w$ 
the operation of each of the selected generators of the group $G$
on the Novikov circle is the map $w\to 1/w$.
This is a composition of the symmetry in the $Ox$ axis
(complex conjugation $w\to\bar{w}$) 
and the inversion $w\to w/|w|^2$
with respect to the circle with center $(0,0)$ and radius 1.
In fact, the inversion preserves each of the two circles.
That's because inversions map circles to circles, 
and it is clear that the points of intersection 
of the Novikov circle 
with the $Ox$ axis, i.e., points $(\sqrt3-\sqrt2,0)$
and $(\sqrt3+\sqrt2,0)$ are conjugated,
because $(\sqrt3-\sqrt2)(\sqrt3+\sqrt2)=1$.
In the sequel, it will be interesting to see
if other components of the self-similar configurations set
for the Novikov circulation are also preserved by the group $G$.

Let us also note that the examples of Novikov
concern actually only one (up to permutation and scaling)
circulation $(1,\lambda,\lambda,1)$.
This follows from the fact that scaling the circulation
does not change the location of collapsing systems.
Permuting the circulation coordinates
corresponds to permuting the coordinates
of the self-similar system.

Novikov circulation has two positive and two negative entries.
O'Neil's theorem proves the existence of collapses
only when one of the configurations is linear rotating, 
and this fact is proved by O'Neil
for circulations $\kappa$,
for which one of the entries $\kappa_l$
has sign opposite to the signs of all remaining entries.
This is not satisfied by the Novikov circulation.
Thus the existence of collapses for the Novikov circulation
does not follow from the O'Neil Theorem.
In the next section we shall show that for this circulation,
components other than the one found by Novikov do exist.

\section{New examples of self-similar systems}
The numerical methods proposed above for hunting 
self-similar rotating and collapsing systems
will be applied now
to producing images of the set of such systems 
for the circulation used by Novikov in his examples.
My calculations suggest that for this circulation,
there are other connected components
of the set of self-similar collapsing and rotating systems,
different from the ellipse found by Novikov.
The calculations were performed by searching
zeros of the vector field $U$ as in Def. \ref{pole_u}.

\begin{figure}[ht]
\caption{Projection to the 2-3 plane.}
\includegraphics[width=0.80\textwidth]{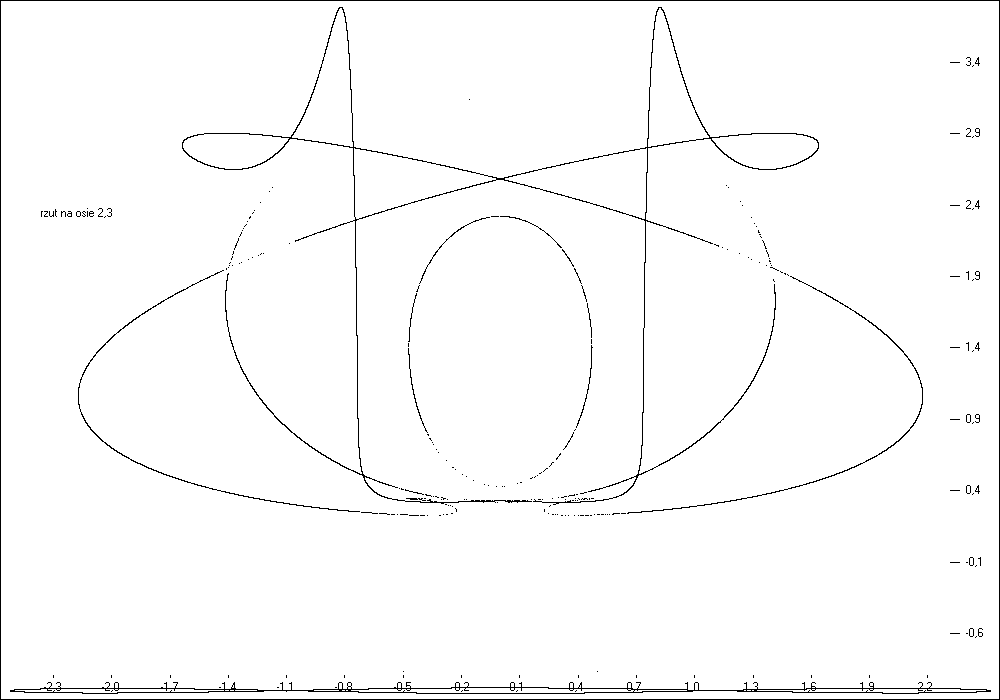}
\end{figure}

\begin{figure}[ht]
\caption{Projection to the 2-4 plane.}
\includegraphics[width=0.80\textwidth]{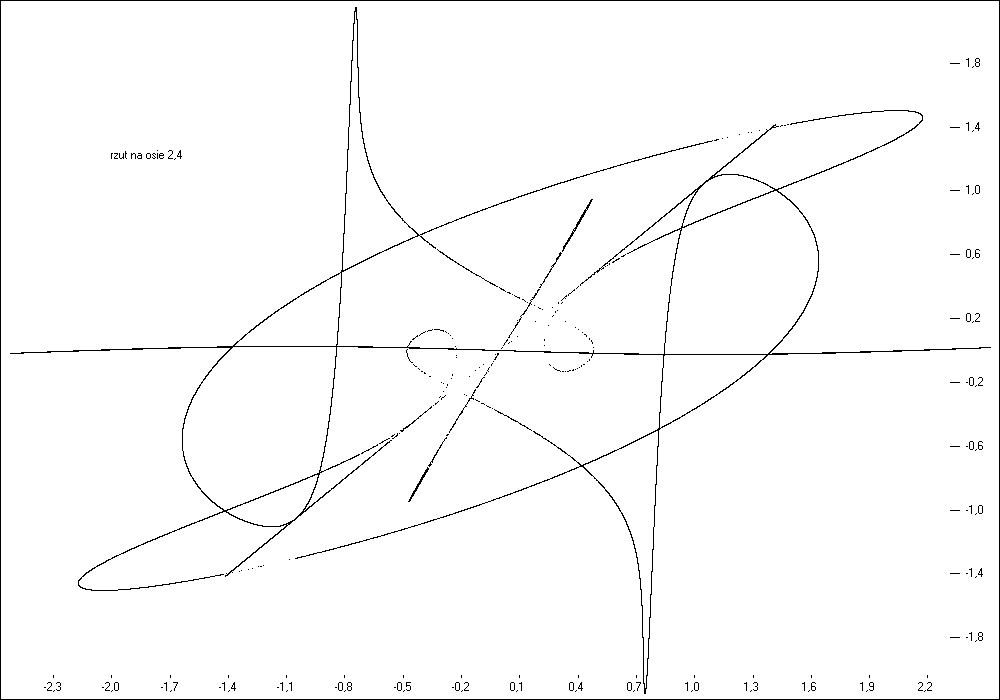}
\end{figure}

\begin{figure}[ht]
\caption{Projection to the 3-4 plane.}
\includegraphics[width=0.80\textwidth]{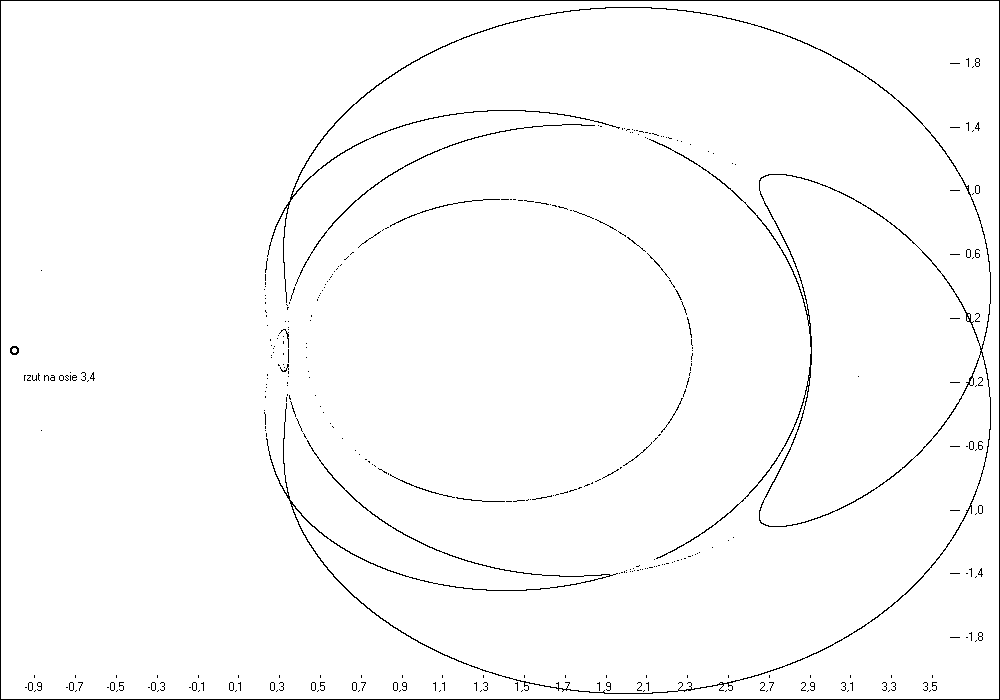}
\end{figure}

Graphic images in Figures 3 to 8
were obtained as described in section \ref{procedura}.
In the four-dimensional cube $[-3.0,+3.0]^4$
we generate random points
and then, using the gradient method,
we lead gradient curves converging to a minimum
of the function $|U(w)|^2$.
If the value found is less than $10^{-8}$,
we believe we found a zero of the vector field $U$,
therefore, a self-similar configuration.
The resulting set of systems is projected 
to the plane coordinates for further analysis.

I present the views of the set $S$ 
of self-similar configurations
projected to various planes of the coordinate system.
The Novikov circle consists of points of $(1+w,w)\in\C^2$,
with $w\in\C$ located on the circle $|w-\sqrt3|^2=2$.
This set can be represented parametrically
$$\theta\to(
1+\sqrt3+\sqrt2\cos\theta,
\sqrt2\sin\theta,
\sqrt3+\sqrt2\cos\theta,
\sqrt2\sin\theta).$$
Thus, the projections to the planes of the coordinate system
will be circles or segments.
In order to easily identify the projection
in the graphic image,
below I list approximate coordinates of the end-points, 
if the projection image is a segment,
and the range of the coordinates, if the projection image is a circle.
\begin{itemize}
	\item $\pi_{12}:\text{ circle },
x_1\in[+1.31784,+4.14626],x_2\in[-1.41421,+1.41421]$
	\item $\pi_{13}:\text{ segment },
L=[+1.31784,+0.317837],P=[+4.14626,+3.14626]$
	\item $\pi_{14}:\text{ circle },
x_1\in[+1.31784,+4.14626],x_4\in[-1.41421,+1.41421]$
	\item $\pi_{23}:\text{ circle },
x_2\in[-1.41421,+1.41421],x_3\in[+0.317837,+3.14626]$
	\item $\pi_{24}:\text{ segment },
L=[-1.41421,-1.41421],P=[+1.41421,+1.41421]$
	\item $\pi_{34}:\text{ circle },
x_3\in[+0.317837,+3.14626],x_4\in[-1.41421,+1.41421]$
\end{itemize}
Here $\pi_{kl}$ means the coordinate system plane
spanned by the axes $1\le k,l\le4$.

In the images presented here,
projections of the five connected 
components can be seen clearly.
Each is a smooth simple closed curve.
Some project to ellipses (or even circles).
The ellipses are flat,
which can be seen since in certain projections 
the images are straight line segments.
Other components are definitely non-planar.

The collection presented here includes about 250,000 points.
Optically one cannot see any significant differences
in comparison to the independently derived
set of only 100 thousand points,
except for a slight improvement in quality.

Some points in the curves escape the region
chosen to be the source of starting points for the procedure.
For example, the maximum of the first coordinate is 3.83.
This exceeds the maximal coordinate of the starting points,
which is 3.0.
This phenomenon is possible because a gradient line
starting in some region may evolve 
and converge even to a distant limit.
Nevertheless, the result is that the limit points are sparse
in some parts of limit curves.
It is advisable to repeat the calculations
with particular emphasis on these areas.

\section{Comparison with O'Neil's results}

The above calculations regarding the collection 
of self-similar systems for the Novikov circulation 
should be compared with the following results of O'Neil \cite{oneil}.
O'Neil proved that collapsing $n$-systems exist for any $n$.
The basis here is Theorem 7.1.1 from p. 411 of his work:
\begin{theorem}
Let $\kappa=(\kappa_1..\kappa_n)$
be a circulation
with one negative entry ($\kappa_n<0$),
and all remaining entries positive: 
$\kappa_l>0$ dla $l=1,\dots,n-1$.
Let us assume that $L=\sum_{l<j}k_lk_j=0$. 
By $s$ denote the number of pairs
$(l,n)$ such that $k_l+k_n>0$.
Then the number of collinear rotating configurations
is at least $s(n-2)!$.
\end{theorem}
Of course, it is easy to find $k_l$ satisfying
these conditions, with $s>0$.
Thus, in this case, collinear rotating configurations do exist.

The existence of collinear rotating configurations
implies the existence of collapsing configurations:
\begin{theorem}
(O'Neil \cite[Theorem 7.4.1, p. 413]{oneil})
Let $n>3$. 
For any circulation
$$\kappa\in K_L^n
=\{k=(k_1,\dots k_n)\in R^n:
k_1=1,\sum k_lk_j=0\}$$
apart from an algebraic subvariety of codimension one, 
each rotating collinear configuration
lies in a one-parameter family
of collapsing configurations.
Each such family is a smooth submanifold
except at finitely many points.
\end{theorem}

It should be noted that the Novikov circulation is of the form
$(1,\kappa,\kappa,1)$,
thus it does not satisfy the assumptions
of the first O'Neil Theorem.
Therefore, the existence of the Novikov examples
is not a consequence of this assertion.
However, some points in the Novikov circle
are collinear rotating configurations.
Therefore the existence of the Novikov circle is,
in some sense,
a consequence of the second O'Neil Theorem.
On the basis of numerical calculations
one can claim that two of the other four components 
of the self-similar configuration set,
exhibited in this work,
contain no collinear rotating configurations.
Thus, in a sense, their existence is predicted by neither the first,
nor the second O'Neil Theorem.

\bibliographystyle{plain}
\bibliography{novikov}

\end{document}